\documentclass[12pt]{article}
\usepackage{amsmath,amssymb,amsthm,graphicx,mathrsfs,color,fancyhdr,fancybox,setspace}
\usepackage{placeins}
\usepackage{amsfonts}
\usepackage[shortlabels]{enumitem}

\newtheorem{theorem}{Theorem}[section]

\newtheorem{conjecture}[theorem]{Conjecture}
\newtheorem{definition}[theorem]{Definition}

  \newtheorem{proposition}[theorem]{Proposition}

\newtheorem{lemma}[theorem]{Lemma}
\newtheorem{corollary}[theorem]{Corollary}

\newcommand{\prooff}[1]{\textit{Proof of Theorem#1.}}

\begin{document}
\thispagestyle{empty}
\begin{center}
\Large
 {\bf Structural properties of biclique graphs and the distance formula}
\vspace*{.5cm}

\large
Marina Groshaus\footnote{Partially supported by CNPq (428941/2016-8). CONICET / Universidad de Buenos Aires, Argentina.}\textsuperscript{,}\footnote{\label{note1}The authors are deeply grateful to the reviewers for their useful and very detailed comments.}\\
 \normalsize
UTFPR-DAINF, Brazil\\
\texttt{marinagroshaus@yahoo.es} \\

\vspace*{.5cm}

\large
Leandro Montero\textsuperscript{\ref{note1}}\\
 \normalsize
KLaIM team, L@bisen, AIDE Lab., Yncrea Ouest\\33 Q, Chemin du Champ de Man\oe{}uvres\\44470 Carquefou, France\\
\texttt{lpmontero@gmail.com}\\
\vspace*{.5cm}

ABSTRACT
\end{center}

\small A \textit{biclique} is a maximal induced complete bipartite subgraph of $G$.
The \textit{biclique graph} of a graph $G$, denoted by $KB(G)$, is the intersection graph of the family of all bicliques of $G$.
In this work we study some structural properties of biclique graphs which are necessary conditions for a graph to be a biclique graph. In particular, we prove that for biclique graphs that are neither a $K_3$ nor a \textit{diamond}, the number of vertices of degree $2$ is less than half the number of vertices in the graph. Also, we present forbidden structures. 
For this, we introduce a natural definition of the distance between bicliques in a graph. We give a formula 
that relates the distance between bicliques in a graph $G$ and the distance between their respective vertices in $KB(G)$. 
Using these results, we can prove not only this new necessary condition involving the degree, but also that some graphs are not biclique graphs. For example, we show that the \textit{crown} is the smallest graph that is not a biclique graph although the known necessary condition for biclique graphs holds, answering an open problem about biclique graphs. Finally, we present some interesting related conjectures and open problems.

\normalsize
\vspace*{.5cm}
\newpage

\section{Introduction}

Intersection graphs of certain special subgraphs of a general graph have been studied 
extensively. We can mention line graphs (intersection graphs of the edges of a graph), interval graphs (intersection graphs
of subpaths of a path), and in particular, clique graphs (intersection graphs of the cliques of a graph) 
\cite{BoothLuekerJCSS1976,BrandstadtLeSpinrad1999,EscalanteAMSUH1973,FulkersonGrossPJM1965, GavrilJCTSB1974, LehotJA1974,McKeeMcMorris1999}.

The \textit{clique graph} of $G$, denoted by $K(G)$, is the intersection graph of the family of all cliques of $G$.
Clique graphs were introduced by Hamelink in~\cite{HamelinkJCT1968} and characterized in~\cite{RobertsSpencerJCTSB1971}. It was proved in~\cite{Alc'onFariaFigueiredoGutierrez2006} 
that the clique graph recognition problem is NP-Complete.

Bicliques have been studied in many contexts. Depending on the context and the author, bicliques are defined in different ways: induced or not induced subgraphs, maximal or not, etc.
(\cite{extra3,extra4,extra5,extra7,extra1,extra6,extra2}).
All of them are rather natural and clearly justified. In our work, we consider bicliques as being maximal induced complete bipartite subgraphs. 

Bicliques have applications in various fields, for example, biology: protein-protein interaction networks~\cite{Bu01052003}, 
social networks: web community discovery~\cite{Kumar}, genetics~\cite{Atluri}, medicine~\cite{Niranjan}, information theory~\cite{Haemers200156}. 
More applications (including some of these) can be found in~\cite{blablamec}.

The \textit{biclique graph} of a graph $G$, denoted by $KB(G)$, is the intersection graph of the family of all bicliques of $G$.
It was defined and characterized in~\cite{GroshausSzwarcfiterJGT2010}. However, no polynomial time algorithm is known for recognizing biclique graphs. Biclique graphs have been studied over the last few years. See~\cite{Coudert2016,grafobic1,grafobic2} for some examples of recent articles on the subject.

In this work we study structural properties of biclique graphs. For this, we introduce the concept of the distance between bicliques in a graph. Previous related work in the context of cliques can be found in~\cite{Balakrishnan,hedman,hedman2,peyrat,PizanaDAM2004}.

In~\cite{GroshausSzwarcfiterJGT2010}, a necessary condition for a graph to be a biclique graph was given. It was an open problem whether this condition was sufficient.  
Using the distances formula, we give a different proof for this necessary condition\footnote{We remark that, although the original proof is short, the new proof shows how to use this formula to solve problems in biclique graphs.}. Moreover, we prove that this necessary condition is not sufficient, that is, we present some structural properties and forbidden structures which allow us to identify graphs that verify the condition but they are not biclique graphs. 
Finally, we present a new necessary condition of biclique graphs: given a biclique graph that is not a $K_3$ or a diamond, we prove that the number of vertices of degree $2$ is 
strictly less than half the number of vertices in the graph. Consequently, we give some forbidden structures. 

We hope that these tools give some light to the main open problem in the context of biclique graphs, that is, the recognition of this class. In the appendix, we present a complete list of all biclique graphs up to $6$ vertices.

This work is organized as follows. In Section $2$ the notation is given. Section $3$ contains some preliminary and general properties. In Section $4$ we present the relation between distances in graphs and biclique graphs plus some structural properties of biclique graphs. Section $5$ contains a bound on the number of vertices of degree $2$ for biclique graphs and forbidden structures.
In the last two sections we present some open and interesting related problems, and we end with a concluding one.

\section{Preliminaries}

Throughout the paper we restrict to undirected simple graphs. Let $G=(V,E)$ be a graph with vertex set $V(G)$ and edge set $E(G)$, and 
let $n=|V(G)|$ and $m=|E(G)|$. A \textit{subgraph} $G'$ of $G$ is a graph $G'=(V',E')$, where $V'\subseteq V$ and 
$E'\subseteq E$ such that all endpoints of the edges of $E'$ are in $V'$. When $E'$ has all the edges of $E$ whose endpoints belong to the vertex subset $V'$, we say that $V'$ \textit{induces} the subgraph $G'=(V',E')$, that is, $G'$ is an \textit{induced subgraph} of $G$. 
A graph $G=(V,E)$ is \textit{bipartite} when there exist sets $U$ and $W$ such that $V= U \cup W$, $U \cap W = \emptyset$, $U \neq \emptyset$, $W \neq \emptyset$ and $E \subseteq U \times W$. Say that $G$ is a 
\textit{complete graph} when every possible edge belongs to $E$. A complete graph on $n$ vertices is denoted $K_{n}$. A bipartite graph is \textit{complete bipartite} when
every vertex of the first set is connected to every vertex of the second set. A complete bipartite graph on $p$ vertices in one set and $q$ vertices in the other is denoted $K_{p,q}$.
A \textit{clique} of $G$ is a maximal complete induced subgraph, while a \textit{biclique} is a maximal induced complete bipartite subgraph of $G$. 
The \textit{open neighborhood} of a vertex $v \in V(G)$, denoted $N(v)$, is the set of vertices adjacent to $v$. 
The \textit{closed neighborhood} of a vertex $v \in V(G)$, denoted $N[v]$, is the set $N(v) \cup \{v\}$.
The \textit{degree} of a vertex $v$, denoted by $d(v)$, is defined as $d(v) = |N(v)|$. A vertex $v\in V(G)$ is \textit{simplicial} if $N[v]$ is a clique. 
A vertex $v\in V(G)$ is \textit{universal} if it is adjacent to all of the other vertices in $V(G)$. A \textit{path} of $k$ vertices, denoted by $P_{k}$, is a sequence of vertices 
$v_{1}v_{2}\ldots v_{k} \in V(G)$ such that $v_{i} \neq v_{j}$ for all $1 \leq i \neq j \leq k$ and $v_{i}$ is adjacent to $v_{i+1}$ 
for all $1 \leq i \leq k-1$. A graph is \textit{connected} if there exists a path between each pair of vertices. The distance between two vertices $v,w \in G$ is defined
as the number of edges in a shortest path between them and is denoted by $d_G(v,w)$. Whenever no confusion arises, we will simply write $d(v,w)$ instead of $d_G(v,w)$. 
We assume that all the graphs of this paper are connected.

A \textit{diamond} is a complete graph with $4$ vertices minus an edge. A \textit{gem} is an induced path with $4$ vertices plus a 
universal vertex.

Given a family of sets $\mathcal{H}$, the \textit{intersection graph} of $\mathcal{H}$ is a graph that has the members of 
$\mathcal{H}$ as vertices and there is an edge between two sets $E,F\in\mathcal{H}$ when $E\cap F \neq \emptyset$.
A graph $G$ is an \textit{intersection graph} if there exists a family of sets $\mathcal{H}$ such that $G$ is the intersection 
graph of $\mathcal{H}$.

\section{General properties}\label{generalProps}

In this section we present some properties of the biclique graph related to connectivity.
First we recall the theorem in~\cite{GroshausSzwarcfiterJGT2010} that gives a necessary condition for a graph to be a biclique graph.

\begin{theorem}[\cite{GroshausSzwarcfiterJGT2010}]\label{tMarina}
Let $G$ be a graph such that $G=KB(H)$, for some graph $H$. Then every induced $P_3$ of $G$ is contained in an induced diamond or an induced gem of $G$ as shown in Figure~\ref{FigtMarina}.
\end{theorem}

 \begin{figure}[ht!]
	\centering
	\includegraphics[trim= 0mm 0mm 0mm 10mm, scale=.4]{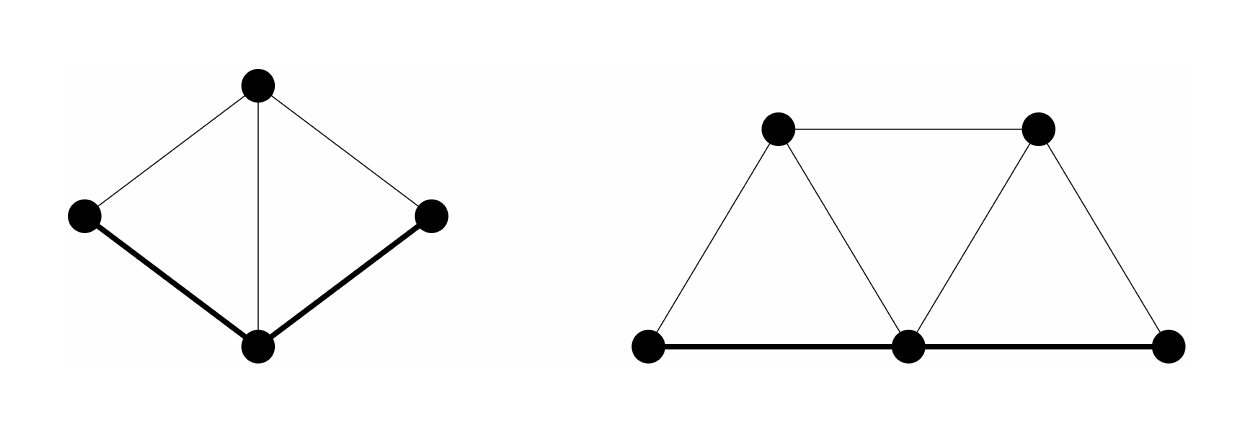}
	\caption{Induced $P_3$ in bold edges contained in a diamond and in a gem respectively.}
	\label{FigtMarina}
\end{figure}

In Section~\ref{sectDist} we give a different proof of Theorem~\ref{tMarina}.
One question that arises from Theorem~\ref{tMarina} is: Given a graph $G$ such that every induced $P_3$ is contained in
a diamond or in a gem; is $G=KB(H)$ for some graph $H$? In Section~\ref{sectDist} we show that the answer is ``No'',
by proving a result that allows us to construct graphs that have every induced $P_3$ in
a diamond or in a gem although they are not biclique graphs. 
    
Next we show the connectivity relation between $G$ and $KB(G)$.

\begin{proposition}\label{p1}
Let $G$ be a graph. $G$ is connected if and only if $KB(G)$ is connected.
\end{proposition}

\begin{proof} Suppose $G$ is connected. Let $B$ and $B'$ be bicliques of $G$.
If $B$ intersects $B'$ then their corresponding vertices in $KB(G)$ are adjacent.
If they do not intersect, as $G$ is connected, there is a path between each vertex of $B$ and $B'$.
Let $b \in B,$ $b' \in B'$ such that $d(b,b') = \min \{d(v,w) $ / $ v \in B$, $w \in B'\}$.
Let $k=d(b,b')$. Clearly $k>0$, so take a path $P=bv_1\ldots v_{k-1}b'$ of length $k$ between $b$ and $b'$.
Now, each triple of consecutive vertices of $P$ is contained in a different biclique since both endpoints of each triple
are not adjacent. Finally taking the bicliques that contain the following triples 
$\{b,v_1,v_2\}, \{v_2,v_3,v_4\},\ldots,\{v_{k-2},v_{k-1},b'\}$ for $k$ even and
$\{b,v_1,v_2\}, \{v_2,v_3,v_4\},\ldots,\{v_{k-3},v_{k-2},v_{k-1}\},\{v_{k-2},v_{k-1},b'\}$ for $k$ odd (note that for $k=1$, the edge $bb'$
is in a biclique that intersects both $B$ and $B'$),
we have that each biclique only intersects with the previous and the following
one. Therefore, their corresponding vertices in $KB(G)$ form a path between the vertices corresponding to $B$ and $B'$. Hence $KB(G)$ is connected.

The converse is clear.
\end{proof}

The following result is a direct consequence of Theorem~\ref{tMarina}.

\begin{lemma}\label{l2}
Let $G$ be a connected graph such that $G=KB(H)$, for some graph $H$. If $G$ has at least $3$ vertices, then $d(v) \geq 2$ for all $v \in V(G)$. Moreover, $G$ is $2$-connected.
\end{lemma}
\begin{proof}
If there is a vertex with degree $1$, since $G$ has at least $3$ vertices, that vertex would be an extreme of a $P_3$ not contained in an induced diamond nor an induced gem,
that is, a contradiction by Theorem~\ref{tMarina}. Finally, suppose $G$ is not $2$-connected and let $v$ be a vertex such that $G - \{v\}$ is disconnected. Taking two vertices $v_1,v_2$ 
in different connected components of $G - \{v\}$ that are adjacent to $v$, we obtain a $P_3$ not contained in an induced diamond nor an induced gem, and again, a contradiction by Theorem~\ref{tMarina}.
\end{proof}

\section{Distances in $G$ and $KB(G)$}\label{sectDist}

In this section we define the distance between bicliques in a graph. Also, we study the relation between the distance
of bicliques in a graph $G$ and the distance between their respective vertices in $KB(G)$.

\begin{definition}\label{defDist}
Let $G$ be a graph and let $B, B'$ be bicliques of $G$. We define the \textit{distance} between $B$ and $B'$ as 
$d(B,B') = \min \{d(b,b') $ /  $b \in B$, $b' \in B'\}$.
\end{definition}

The next formula states the relationship between the distances of $G$ and $KB(G)$. 
This result is useful to show that the condition of Theorem~\ref{tMarina} is not sufficient. 

\begin{lemma}\label{l4}
Let $G$ be a graph and let $B, B'$ be two different bicliques of $G$. Then 
$d_{KB(G)}(B,B') = \big\lfloor\frac{d_G(B,B') + 1}{2} \big\rfloor + 1$.
\end{lemma}

\begin{proof} Let $v_0,v_k$ be vertices of $G$ such that $v_0 \in B$, $v_k \in B'$ and $d(v_0, v_k) = d_G(B, B') = k$.
If $k=0$ then $B$ and $B'$ intersect in $G$. So they are adjacent as vertices in $KB(G)$. Therefore
$d_{KB(G)}(B, B') = \big\lfloor\frac{0 + 1}{2} \big\rfloor +1= 1$. Suppose now that $k>0$. Let $P_1=v_0 v_1\ldots v_k$ 
be a path in $G$ between $B$ and $B'$ of length $k$. Take $B_i \in V(KB(G))$ such that
$\{v_i, v_{i+1}, v_{i+2}\} \subseteq B_i$ in $G$, $i=0,\ldots,k-2$. These bicliques $B_i$ of $G$ exist since  
$v_i v_{i+2} \notin E(G)$ for $i=0,\ldots,k-2$, otherwise there would be a path of length less than $k$ between $B$ and $B'$.

Then $B B_0 B_2 B_4\ldots B_{2j} \ldots B'$ is a path in $KB(G)$ between $B$ and $B'$ of length
$\big\lfloor\frac{k + 1}{2} \big\rfloor + 1$ and therefore, as $d_G(B,B')=k$, we have that 
$d_{KB(G)}(B,B')$ $\leq \big\lfloor\frac{d_{G}(B,B') + 1}{2} \big\rfloor + 1$.
This situation can be observed in Figure~\ref{Fig20}.

\begin{figure}[ht!]
  \centering
  \includegraphics[scale=.5]{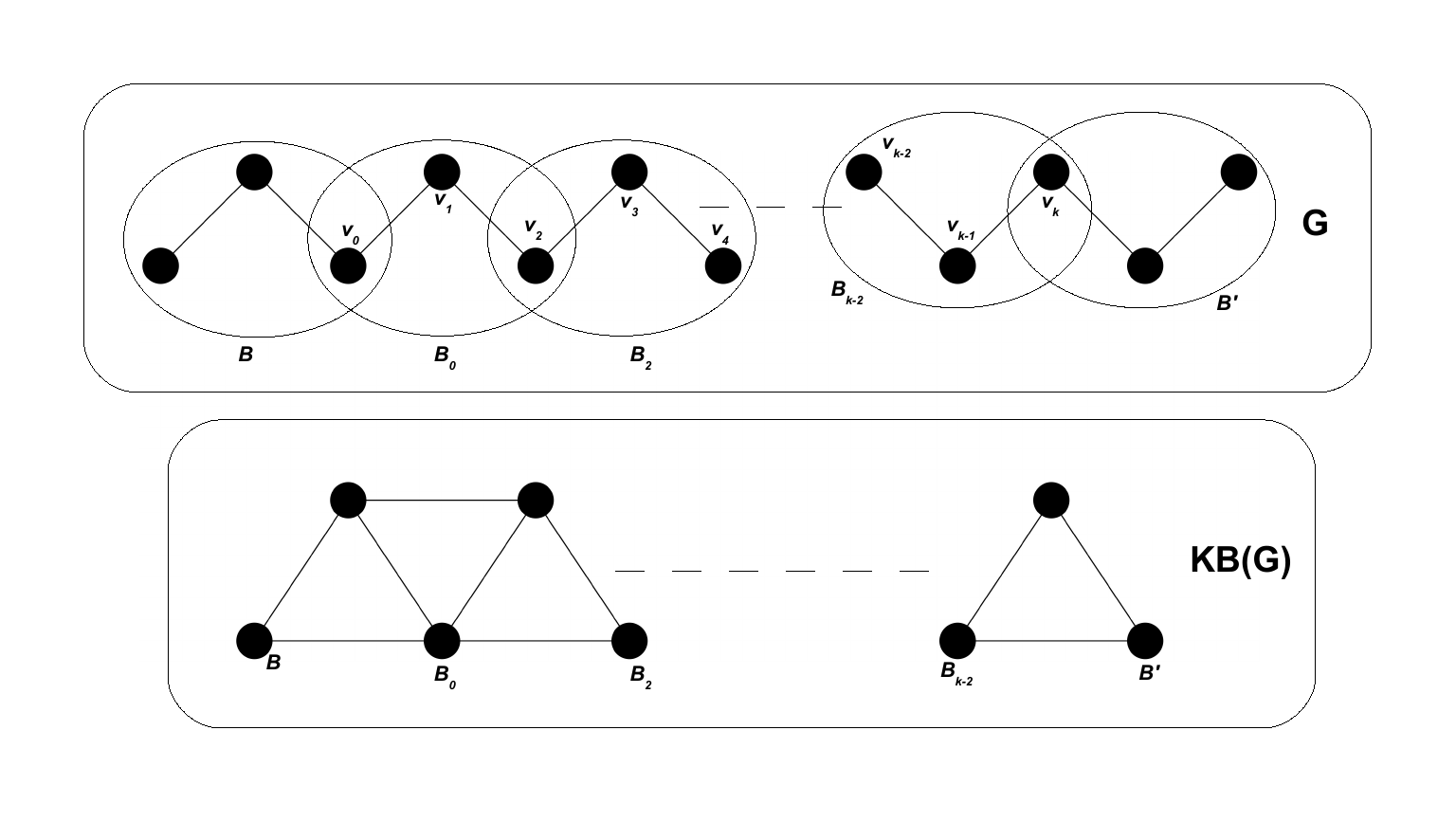}
  \caption{First inequality.}
  \label{Fig20}
\end{figure}

Now let $P_2=B_0 B_1\ldots B_s$ be a path of minimum length in $KB(G)$ between $B=B_0$ and $B'=B_s$ (Fig.~\ref{Fig21}). 
Then $d_{KB(G)}(B, B') = s > 0$. Let $v_{2i} \in B_i \cap B_{i+1}$ be vertices of $V(G)$ for $i=0,\ldots,s-1$.
Thus, we obtain the vertices $v_0, v_2,\ldots,v_{2s-2}$ of $G$. Now, for $i=1,\ldots,s-1$, either $v_{2i-2}$ is adjacent 
to $v_{2i}$ or $v_{2i-2}$ is not adjacent to $v_{2i}$. If they are not adjacent then there exists one vertex $v_{2i-1}$ adjacent
to both since $v_{2i-2}$ and $v_{2i}$ belong to the biclique $B_{i}$ of $G$. Then, the longest path between $v_0$ and 
$v_{2s-2}$ occurs when these two consecutive vertices are not adjacent. In this situation, adding the vertex
adjacent to both between each pair, we have an induced path $v_0v_1\ldots v_{2s-2}$ in $G$ between $B$ and $B'$ of length $2s-2$. 
Then,

\begin{displaymath}
d_G(B,B') \leq 2s-2 
\end{displaymath}

\begin{displaymath}
\frac{d_G(B,B')+2}{2} \leq s = d_{KB(G)}(B,B')
\end{displaymath}

\begin{displaymath}
\frac{d_G(B,B')}{2} + 1 \leq d_{KB(G)}(B,B')
\end{displaymath}

Finally, since $d_G(B,B')$ and $d_{KB(G)}(B,B')$ are integers, it follows that

\begin{displaymath}
\Big\lfloor\frac{d_G(B,B') + 1}{2} \Big\rfloor + 1 \leq d_{KB(G)}(B,B')
\end{displaymath}

Combining both inequalities we obtain the desired result.

\begin{figure}[ht!]
  \centering
  \includegraphics[scale=.57]{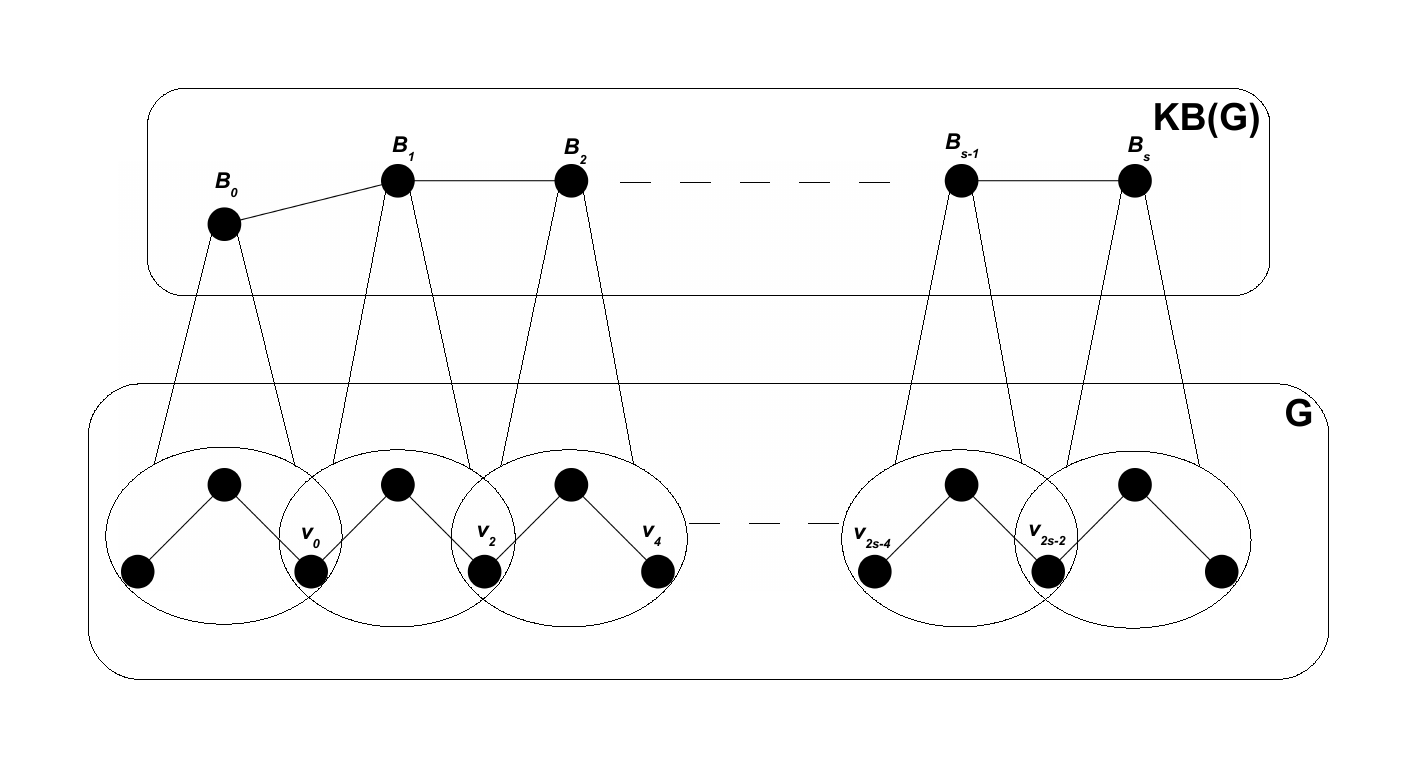}
  \caption{Second inequality.}
  \label{Fig21}
\end{figure}

\end{proof}

Now, based on the distance between two bicliques of a graph $G$, we can ensure the existence of other bicliques ``between them''.
That is, if the distance between the bicliques $B$ and $B'$ of $G$ is $k$, then there exist other bicliques at distance at most
$k-1$ to each of $B$ and $B'$.
This result will be very useful for proving not only Theorem~\ref{tMarina} but also that the condition of the theorem is not sufficient.

\begin{theorem}\label{teodistk}
Let $G$ be a graph and let $B, B'$ be bicliques of $G$ such that $d_G(B, B')=k > 0$. Then there exist at least $k+1$ bicliques in $G$ such that they are at distance at most $k-1$ from both $B$ and $B'$.
\end{theorem}

\begin{proof} We will prove the theorem in two parts, first when $d_G(B, B')=1$ and last when $d_G(B, B')=k > 1$.

Suppose first that $d_G(B, B')=1$. Let $v \in B,$ $w \in B'$ be adjacent vertices and let $x \in B,$ $y \in B'$ be vertices such that $vx,wy \in E(G)$. 
In Figure~\ref{Fig22} we show all possible configurations (up to symmetry) according to the adjacencies between vertices. We count the number of different bicliques in each case.\\

\begin{figure}[ht!]
  \centering
  \includegraphics[trim = 0mm 10mm 0mm 20mm, scale=.25]{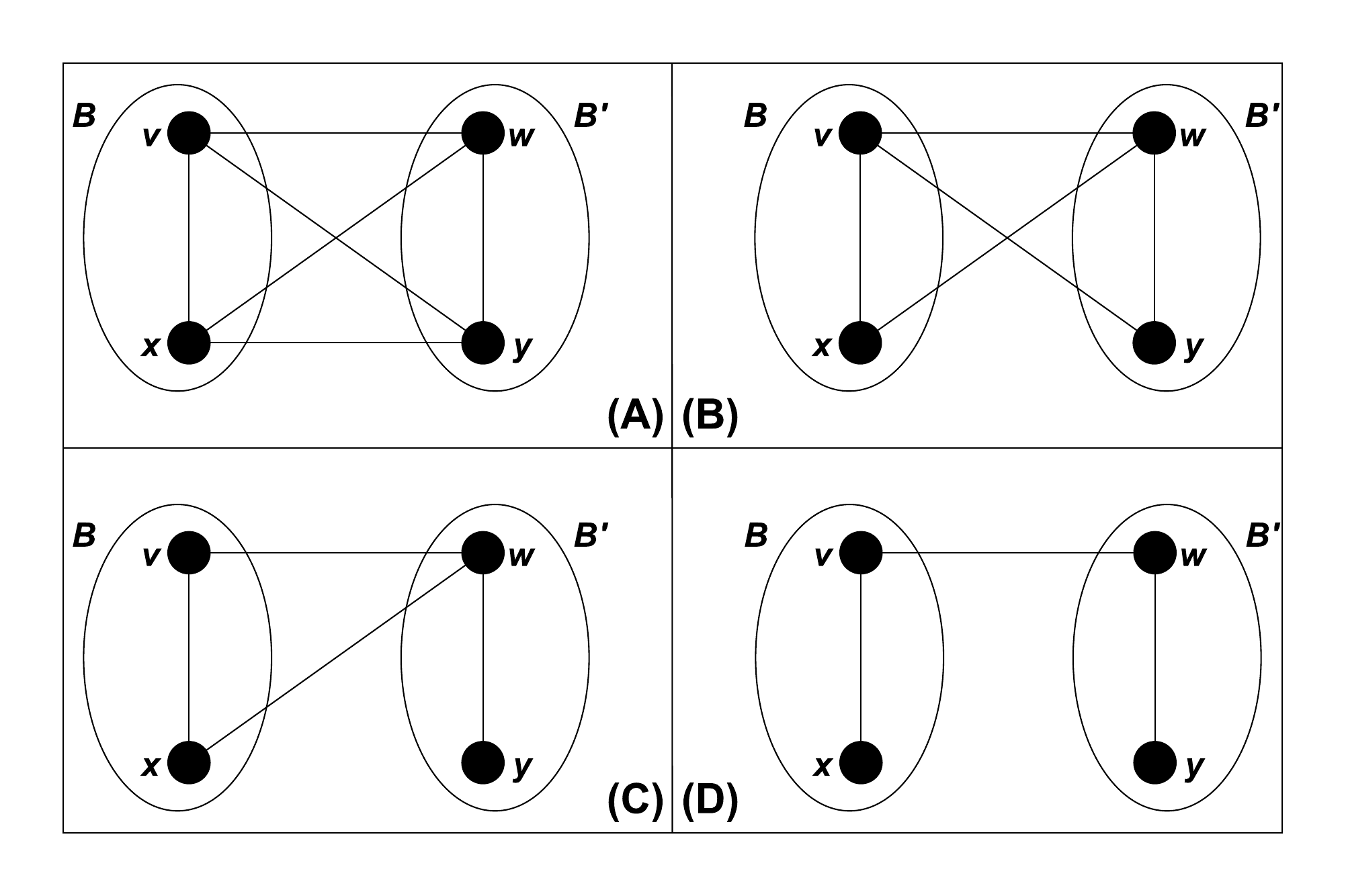} 
  \caption{Different cases for bicliques $B$ and $B'$ at distance $1$.}
  \label{Fig22}
\end{figure}

\textbf{(A)} $\{x,y\}, \{v,y\}, \{v,w\}, \{x,w\}$ are contained in four different bicliques.\\

\textbf{(B)} $\{x,v,y\}, \{x,w,y\}$ and $\{v,w\}$ are contained in three different bicliques.\\

\textbf{(C)} $\{x,w,y\}$ and $\{v,w,y\}$ are contained in two different bicliques.\\

\textbf{(D)} $\{x,v,w\}$ and $\{v,w,y\}$ are contained in two different bicliques.\\

Note that if $xy \in E(G)$ and $xw, vy \notin E(G)$, then $\{x,v,w,y\}$ is contained in a biclique.
Now, as the biclique containing $\{v,w,x,y\}$ should be different from $B$, there exists a vertex $z \in B$ such that either $zv,zw \in E(G)$ (or $zx, zy \in E(G)$) or $zv \in E(G)$, $zy \notin E(G)$ (or $zx \in E(G)$, $zw \notin E(G)$). 
By symmetry, as the biclique containing $\{v,w,x,y\}$ is different from $B'$, there should also exist a vertex $u \in B'$ with similar adjacencies.
In these situations, we can see that we fall in one of the previous cases \textbf{(B)}, \textbf{(C)} or \textbf{(D)}, therefore taking
into account the biclique containing $\{v,w,x,y\}$, we obtain at least three different bicliques.

In all cases there are at least two different bicliques that intersect $B$ and $B'$, that is, they are at distance $k-1=0$ to each of them. We remark that we exactly count in each case the minimum number of bicliques that can exist between $B$ and $B'$ as we will use that later in the paper.

 \begin{figure}[ht!]
  \centering
  \includegraphics[trim = 10mm 10mm 0mm 10mm, scale=.47]{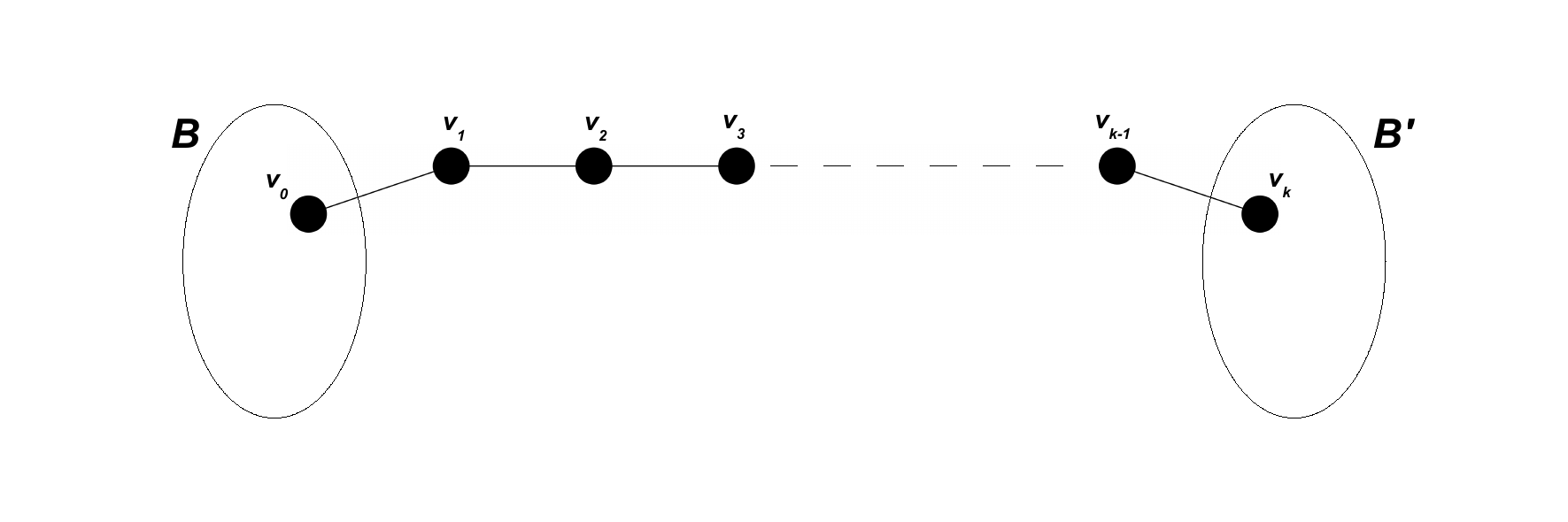}
  \caption{Bicliques $B$ and $B'$ are at distance $k$.}
  \label{Fig23}
\end{figure}

Suppose last that $d_G(B,B')=k > 1$ and let $v_0v_1\ldots v_k$ be a shortest path between $B$ and $B'$ such that $v_0 \in B,$ $v_k \in B'$.
Then, each triple $\{v_{i},v_{i+1},v_{i+2}\}$ is contained in a different biclique of $G$ for $i=0,\ldots,k-2$.
Therefore we obtain $k-1$ bicliques that are at distance at most $k-1$ to each of $B$ and $B'$. 
We obtain the two remaining bicliques as follows. Let $x\in B,$ $y \in B'$ such that $xv_0,yv_k \in E(G)$.
If $xv_1 \notin E(G)$ (respectively $yv_{k-1} \notin E(G)$) there is a biclique containing $\{x,v_0,v_1\}$ (respectively $\{y,v_k,v_{k-1}\}$).
We call this biclique sharing an edge with $B$ (respectively $B'$) \textit{special biclique}. Otherwise, if $xv_1 \in E(G)$ (respectively $yv_{k-1} \in E(G)$),
then $\{x,v_1,v_2\}$ (respectively $\{y,v_{k-1},v_{k-2}\}$) is contained in a biclique.
\end{proof}

As an immediate result of Theorem~\ref{teodistk} for $k=1$ and $k=2$ we obtain the following two corollaries.

\begin{corollary}\label{coroLoco1}
Let $G$ be a graph, let $B, B'$ be bicliques of $G$ such that $d_G(B,B')=1$, let $e$ be an edge
with one endpoint in $B$ and the other in $B'$, and let $B_2$ be a biclique containing $e$. Then, there exists a biclique $B_1$
such that it intersects $B,B'$ and $B_2$. Moreover, the vertices $b,b',b_1,b_2 \in KB(G)$ corresponding to the bicliques $B,B',B_1,B_2$ in $G$ respectively, induce a diamond in $KB(G)$.
\end{corollary}

\begin{corollary}\label{coroLoco2}
Let $G$ be a graph, let $B, B'$ be bicliques of $G$ such that $d_G(B,B')=2$. Consider an induced $P_{3}$ with its extremes in $B$ and $B'$, and let $B_2$ be a biclique containing that $P_3$.
Then, there exist two different bicliques $B_1,B_3$ ``between'' $B$ and $B'$, such that calling
$b,b',b_1,b_2,b_3$ the vertices in $KB(G)$ corresponding to the bicliques $B,B',B_1,B_2,B_3$ in $G$ respectively, we either have  
\begin{itemize}
\item that both $B_1$ and $B_3$ are special bicliques, thus $\{b,b',b_1,b_2,b_3\}$ induces a gem in $KB(G)$, or
\item that at least one of $B_1$ and $B_3$ is not special (suppose $B_1$ is not), thus $\{b,b',b_1,b_2\}$ induces a diamond in $KB(G)$.
\end{itemize}
\end{corollary}

\vspace*{3mm}
The proof of Theorem~\ref{tMarina} follows from Theorem~\ref{teodistk}, and Corollaries~\ref{coroLoco1} and~\ref{coroLoco2}.
\vspace*{3mm}

\noindent\prooff{~\ref{tMarina}} Let $bb_2b'$ be an induced $P_{3}$ in $G$ and let $B$, $B_2$ and $B'$ be the bicliques of $H$ associated to the vertices 
$b$, $b_2$ and $b'$ of $G$.
The biclique $B_2$ contains either an edge with one endpoint in $B$ and the other in $B'$, or $B_2$ contains a $P_3$ with its extremes in $B$ and $B'$. 
In the first case, by Corollary~\ref{coroLoco1}, we obtain that $bb_2b'$ is contained in an induced diamond of $G$.
In the second case, by Corollary~\ref{coroLoco2}, we obtain that $bb_2b'$ is contained in either an induced diamond or an induced gem of $G$.
\qed

\vspace*{5mm}

Now we will show that although in the \textit{crown} (Fig.~\ref{FigExtra3}) every induced $P_{3}$ is contained in an induced diamond, it is not a biclique graph. This result is a counterexample of the question of the sufficiency of the property of Theorem~\ref{tMarina}.
\begin{figure}[ht!]	
  \centering	
  \includegraphics[scale=.25]{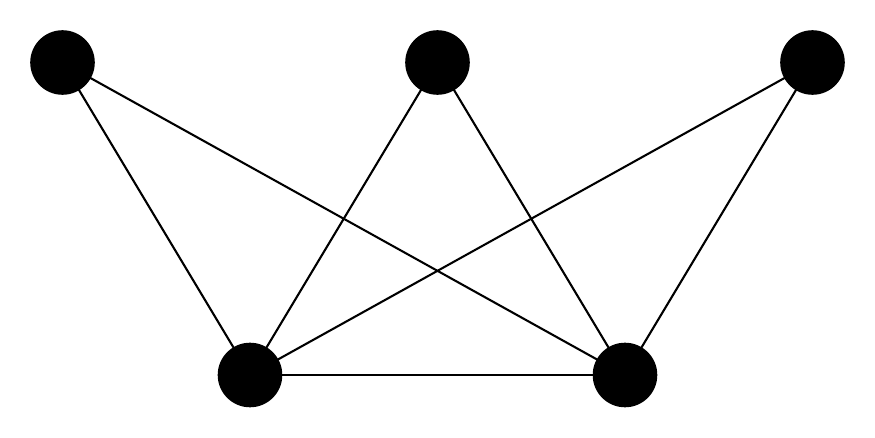} 	
  \caption{The \textit{crown} is not a biclique graph but has every $P_{3}$ in a diamond.}
  \label{FigExtra3}	
\end{figure}

Indeed, we will prove a more general result that implies not only that the \textit{crown} graph is not a biclique graph but neither are many other graphs. We remark that the \textit{crown} is the smallest graph that verifies the condition of Theorem~\ref{tMarina} but is not a biclique graph.

\begin{proposition}\label{pULTI}
Let $G=KB(H)$ for some graph $H$, where $G$ is not isomorphic to the diamond. Then, there do not exist $v_1,v_2 \in V(G)$ such that 
$N(v_1)=N(v_2)$ and their neighbors induce a $K_2$.
\end{proposition}

\begin{proof} Suppose that there exist $v_1,v_2 \in V(G)$ such that $N(v_1)=N(v_2)$ and their neighbors induce a $K_2$. Then, 
$d_{G}(v_1,v_2)=2$ and therefore, if $B$ is the biclique of $H$ that corresponds to $v_1$ and $B'$ is the biclique of $H$ 
that corresponds to $v_2$, by Lemma~\ref{l4}, $d_{H}(B,B')=2$ or $d_{H}(B,B')=1$. We will analyze each case.

\begin{itemize}
\item Case $d_{H}(B,B')=2$. Now, $H$ must contain a subgraph as depicted in Figure~\ref{FigUltiPosta}. We will show that we arrive at a contradiction.
Suppose first that one of both dotted edges does not exist, say $vy$. Then $\{v,x,y\}$ is contained in a biclique that does not intersect
$B'$. This is a contradiction since $N(v_1)=N(v_2)$ in $G$.
Suppose next that both dotted edges $vy$ and $vy'$ exist. In this case we also arrive at a contradiction since in $H$ there are at least four bicliques 
that intersect with $B$ and $B'$. We obtain one for each choice of one element in $\{x,y\}$ in $B$, one element in $\{x',y'\}$ in $B'$ and $v$. 
Therefore $N(v_1)=N(v_2)$ does not induce a $K_2$, which is a contradiction.

\begin{figure}[ht!]
  \centering	
  \includegraphics[scale=.3]{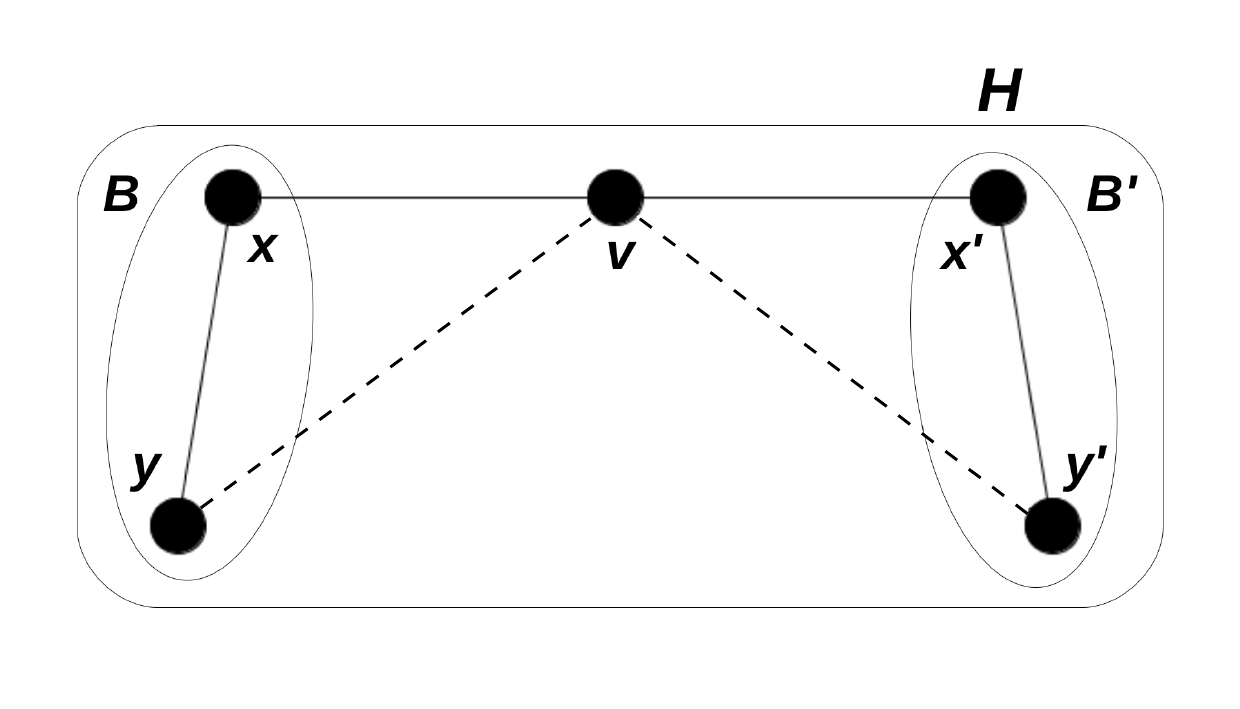} 	
  \caption{Graph $H$ when $d_{H}(B,B')=2$.}
  \label{FigUltiPosta}	
\end{figure}

  \item Case $d_{H}(B,B')=1$.
In this case, by Theorem~\ref{teodistk}, there exist at least two bicliques $B_{1},B_{2}$ in $H$ that intersect both $B$ and 
$B'$. Clearly they must be exactly two, otherwise $N(v_1)=N(v_2)$ would not induce a $K_{2}$. Now, by the first part of the proof of Theorem~\ref{teodistk}, only in the cases \textbf{(C)} and 
\textbf{(D)} it is possible that there are exactly two bicliques intersecting both $B$ and $B'$ in $H$. Figure~\ref{FigDOSCASOS} shows both possible options. 
The labels of Theorem~\ref{teodistk} were as given in Figure~\ref{Fig22}. 

\begin{figure}[ht!]	
  \centering	
  \includegraphics[scale=.4]{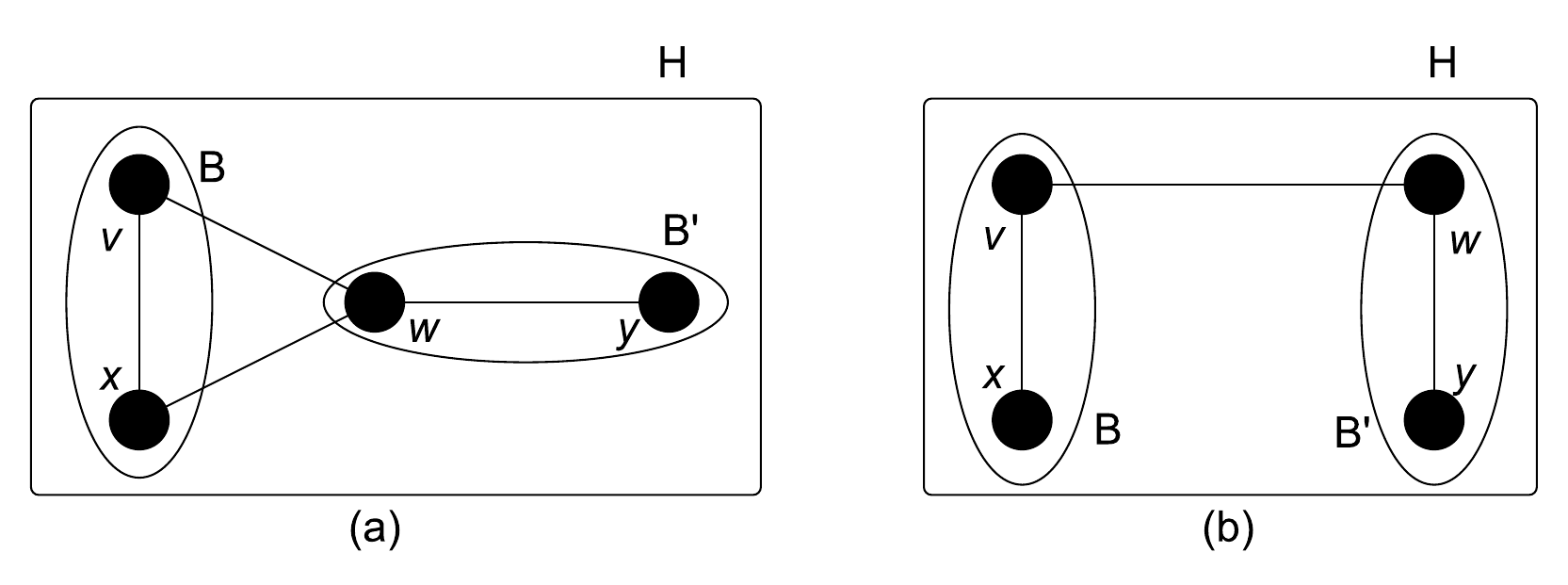} 	
  \caption{Only options for $H$ with two bicliques that intersect $B$ and $B'$.}
  \label{FigDOSCASOS}	
\end{figure}

    $-$ Case \textbf{(a)} (follows case \textbf{(C)} of the proof of Theorem~\ref{teodistk}):    
     
    As we can observe in Figure~\ref{FigDOSCASOS}\textbf{a}, if $B_1$ is the biclique containing $\{v, w, y\}$ and $B_2$
is the biclique containing $\{x, w, y\}$, $H$ has four bicliques such that they induce a 
diamond in $G$. As $G$ is not isomorphic to the diamond, there must exist another biclique in $H$ that intersects neither $B$ nor $B'$. Let $B_3$ be another biclique such that $V(B_3) \subseteq V(H) \setminus (B \cup B')$.
Suppose now that there is a vertex $u \in B_3$ such that $u \notin (B_1 \cup B_2)$. Since $H$ is connected, we can
choose $B_3$ such that there is a vertex $u' \in (B \cup B' \cup B_1 \cup B_2)$ adjacent to $u$. Clearly, if $u' \in (B \cup B')$,
we obtain a contradiction as the edge $uu'$ is contained in a biclique different to $B_1$, $B_2$ that intersects $B$ or $B'$.
Otherwise, since the edge $wy \in (B_1 \cap B_2)$ and $u' \in (B_1 \cup B_2)$, we have that $u'$ is adjacent to $w$ or $y$, therefore
we obtain that $\{u,u',w\}$ or $\{u,u',y\}$ is contained in another biclique intersecting $B'$ respectively which is a contradiction.
We can conclude that $V(H) \setminus (B \cup B')$ is contained in $B_1 \cup B_2$.

We show now that there are at least three bicliques that intersect $B$ or $B'$, i.e., $v_1$ or $v_2$ has an open neighborhood bigger than two vertices what would be a contradiction. For this let $ab$ be an edge in the biclique $B_3$. 
As $V(B_3)$ is contained in $B_1 \cup B_2$, both $a$ and $b$ are adjacent to at least one vertex in $\{w, y\}$.

If $a$ is adjacent to both $w$ and $y$, then the edges $ab$ in $B_3$ and $wy$ in $B'$
determine case \textbf{(A)} or \textbf{(B)} of the proof of Theorem~\ref{teodistk}, which is a contradiction. Therefore, $a$ is adjacent
to just one vertex in $\{w,y\}$. Similarly, the same can be said about $b$.

If $a$ and $b$ are adjacent to different vertices in $\{w,y\}$, then these four
vertices induce a $C_4$ and hence there are three bicliques between $B_3$ and $B'$
like in the proof of Theorem~\ref{teodistk}, a contradiction.

If $a$ and $b$ are both adjacent to $y$, consider the bicliques containing
$\{w, y, a\}$ and $\{w, y, b\}$. They cannot be the same biclique, so one is $B_1$ and
the other is $B_2$. Suppose without loss of generality that the first biclique
is $B_1$ and the second biclique is $B_2$. Recalling that $B_1$ is the biclique containing 
$\{v, w, y\}$ and $B_2$ is the biclique containing $\{x, w, y\}$, one can conclude
that $a$ is adjacent to $v$ and that $b$ is adjacent to $x$. Then, if we consider the
edge $vx$ in $B$ and the edge $ab$ in $B_3$, we have that $\{a, b, v, x\}$ either induces
a $C_4$, a diamond (case \textbf{(B)}) or a $K_4$ (case \textbf{(A)}). Thus there are at least
three bicliques between $B$ and $B_3$, a contradiction.

Finally, if both $a$ and $b$ are adjacent to $w$, this time suppose without loss
of generality that $B_1$ is the biclique containing $\{a, w\}$ and that $B_2$ is the
biclique containing $\{b, w\}$. If $v$ is not adjacent to $b$, then any biclique containing 
$\{v, w, b\}$ is different from $B_1$ and $B_2$, so there are three bicliques that
intersect $B'$, a contradiction. Consequently, $v$ must be adjacent to $b$. Similarly,
$x$ is adjacent to $a$. Therefore, when we consider the edges $vx$ in $B$ and $ab$ in $B_3$, we see that $\{a, b, v, x\}$ induces a $C_4$, diamond or $K_4$, so there are at least three bicliques between $B$ and $B_3$, a contradiction.

$-$ Case \textbf{(b)} (follows case \textbf{(D)} of the proof of Theorem~\ref{teodistk}): 

As we can observe in Figure~\ref{FigDOSCASOS}\textbf{b}, if $B_1$ is the biclique containing $\{x,v,w\}$ and $B_2$
is the biclique containing $\{v,w,y\}$, in a similar way of case \textbf{(a)}, $H$ has four bicliques that induce a diamond in $G$. Then, there must exist other biclique $B_3$ in $H$ different from these four such that it does not intersect $B$ or $B'$.

One can show that $V(H) \setminus (B \cup B')$ is contained in $B_1 \cup B_2$ with exactly the same argument as previous case
using now that the edge $vw \in (B_1 \cap B_2)$ rather that the edge $wy \in (B_1 \cap B_2)$ as we did in case \textbf{(a)}.

Now we show for this case that there are at least three bicliques that intersect $B$ or $B'$. 
Let $B_3$ be a fifth biclique in $H$ not intersecting $B$ or $B'$ and let $ab$ be an edge in $B_3$. As $V(B_3)$ is contained in $B_1 \cup B_2$, both $a$ and $b$ are adjacent to at least one vertex in $\{v,w\}$.

If $a$ is adjacent to both $v$ and $w$, then the bicliques containing edges $av$ and $aw$ are different from $B_1$ and $B_2$, which is a contradiction. Therefore, $a$ is adjacent to just one vertex in $\{v,w\}$ and the same can be said about $b$.

If $a$ and $b$ are adjacent to different vertices in $\{v,w\}$, suppose without loss of generality that $a$ is adjacent to $v$ and
$b$ adjacent to $w$, then these four vertices induce a $C_4$ that is contained in a biclique $\tilde{B}$. If $\tilde{B}$ is different from $B_1$ and $B_2$, we obtain a contradiction. Otherwise $\tilde{B}$ should be $B_1$ or $B_2$. Without loss of generality, suppose that
$\tilde{B} = B_1$, therefore $x$ must be adjacent to $b$. Considering the edge $ab$ in $B_3$ and $xv$ in $B$, $\{a,b,x,v\}$ induces
a $C_4$ and hence there are three bicliques between $B_3$ and $B$ as in the proof of Theorem~\ref{teodistk}, a contradiction.

Finally, if $a$ and $b$ are both adjacent to $v$ (symmetric if both are adjacent to $w$), we have that either $\{x,v,a\}$ and $\{x,v,b\}$
are contained in bicliques different from $B_2$ which is a contradiction, or $\{a,b,x,v\}$ induces a diamond or a $K_4$ and
therefore there are at least three bicliques between $B$ and $B_3$ which is also a contradiction.

\end{itemize}

As no more cases are left, there do not exist $v_{1},v_{2} \in V(G)$ such that $N{(v_{1})}=N{(v_{2})}$ with their neighbors inducing a $K_2$ which completes the proof.

\end{proof}

Figure~\ref{FigExtraExtra3} shows some examples of graphs where every $P_{3}$ is included in a diamond that are not biclique graphs.

\begin{figure}[ht!]
  \centering	
  \includegraphics[scale=.5]{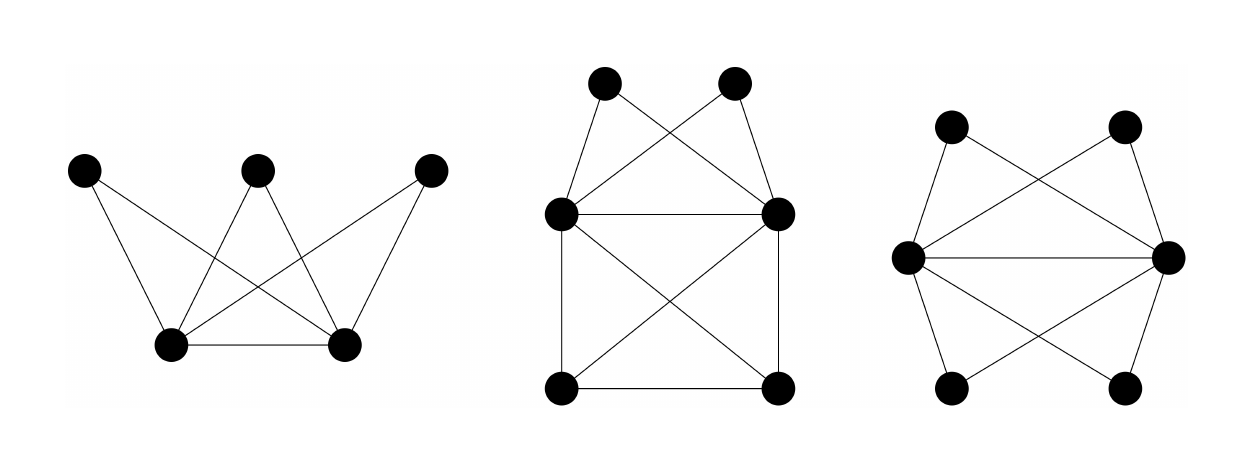} 	
  \caption{Graphs that are not biclique graphs by Proposition~\ref{pULTI}.}
  \label{FigExtraExtra3}
\end{figure}

\section{Vertices of degree two in biclique graphs}

In this section we give a strong property for biclique graphs that have an induced $P_3$ contained in a gem and not in 
a diamond. Also, we show some forbidden structures.
These properties give more tools to recognize graphs that are not biclique graphs. 

The next result implies that the \textit{Haj\'os graph}, the \textit{rising sun} and the $X_1$ graph (see Fig.~\ref{hajosandsun}) are not biclique graphs by giving a forbidden structural property.

\begin{proposition}\label{B3} 
Let $G=KB(H)$ for some graph $H$, let $bb_2b'$ be an induced $P_3$ such that $b,b'$ do not belong to any induced diamond. Let $b_1,b_3$ be the vertices of $G$ corresponding to special bicliques in $H$ such that $\{b,b',b_1,b_2,b_3\}$ induces a gem (Corollary~\ref{coroLoco2}). If $\tilde{b}$ is a vertex such that $\tilde{b}\notin \{b,b',b_1,b_2,b_3\}$
and $\tilde{b}$ does not belong to an induced diamond with $b$, then $\tilde{b}$ is not adjacent to $b_1$.
\end{proposition}
\begin{proof}
Recall that since $b,b'$ do not belong to any induced diamond and are at distance $2$, by Corollary~\ref{coroLoco2}, such vertices $b_1,b_3$ exist.
Let $B,B',B_1,B_2$, $B_3,\tilde{B}$ be the bicliques of $H$ corresponding to the vertices $b,b',b_1,b_2,b_3,\tilde{b}$ of $G$.
Let $xy$ be an edge that belongs to $B\cap B_1$. Note that this edge exists since $B_1$ is a special biclique.
By contradiction, if $\tilde{B}$ and $B_1$ intersect (i.e., $\tilde{b}$ is adjacent to $b_1$ 
in $G$), then either $x$ is adjacent to some vertex $z$ of $\tilde{B}$, or the vertex $y$ is adjacent to some vertex $z$ of $\tilde{B}$. 
In either case, there is an edge between a vertex of $B$ and a vertex of $\tilde{B}$, which is a contradiction by 
Corollary~\ref{coroLoco1}, as $\tilde{b}$ and $b$ would belong to an induced diamond in $G$.
\end{proof}

\begin{corollary} 
The Haj\'os graph, the rising sun and the $X_1$ graph are not biclique graphs (Fig.~\ref{hajosandsun}).
\end{corollary}

Moreover, these three graphs give forbidden structures for biclique graphs.

\begin{corollary}\label{coroHajos}
Let $G$ be a graph which contains the \textit{Haj\'os graph}, the \textit{rising sun} or the $X_1$ graph as induced subgraph where the vertices of degree $2$ in the subgraph are also of degree $2$ in $G$. Then, $G$ is not a biclique graph. 
\end{corollary}

\begin{figure}[ht!]
  \centering
  \includegraphics[scale=.3]{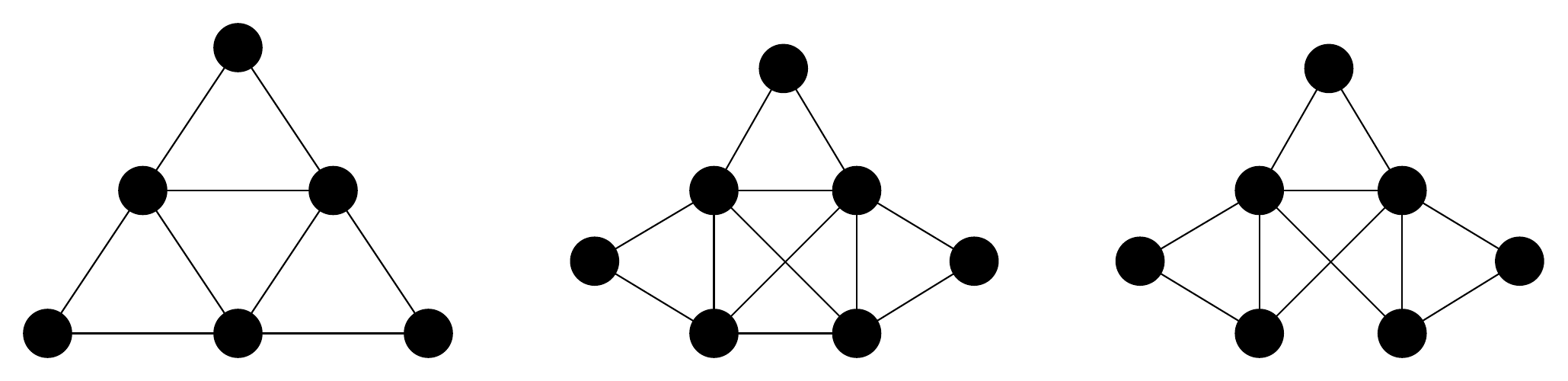} 
  \caption{The \textit{Haj\'os graph}, the \textit{rising sun} and the $X_1$ graph.}
  \label{hajosandsun}
\end{figure}

Next we present the theorem that gives an upper bound on the number of vertices of degree $2$ in a biclique graph.

\begin{theorem}\label{teogrado2}
Let $G=KB(H)$ for some graph $H$, where $G$ is not isomorphic to $K_3$ or the diamond, and $|V(G)|=n$. Then, the number of vertices of degree $2$ in $G$ is strictly less than $n/2$.
\end{theorem}

\begin{proof}
Let $V_2 = \{ v \in G : d(v) = 2\}$. First we show that there are no edges between the vertices of $V_2$. 
Suppose by contradiction that $v_i,v_j\in V_2$ are adjacent. If they have a common neighbor, say $w$, since $G$ is not isomorphic to $K_3$, both 
$\{v_i,w,w'\}$ and $\{v_j,w,w'\}$, where $w'$ is any other neighbor of $w$, induce a $P_3$ which is not contained in a diamond or a gem. Otherwise, if the vertex $w$ is adjacent to $v_i$ and not adjacent to $v_j$, then $\{w,v_i,v_j\}$ induces a $P_3$ 
that is not included in a diamond or a gem since no other vertex is adjacent to $v_i$. In both cases we arrive at a contradiction by Theorem~\ref{tMarina}.

Next we show that for each $v_i \in V_2$, there exists a vertex $w_i \in N(v_i)$ such that $w_i \notin N(v_j)$, for all $v_j \in V_2$, $j \neq i$. Clearly, $w_i\notin V_2$.
By contradiction, suppose that there is a vertex $v_i \in V_2$ such that each of
its two neighbors, say $v', v''$, are adjacent to some other vertex of $V_2$. Note again that by Theorem~\ref{tMarina}, $v'$ and $v''$ are adjacent. Consider the following cases:
\begin{itemize}
 \item $N(v_i) = N(v_j) = \{v',v''\}$, for some $v_j \in V_2$. Since $G$ is not isomorphic to the diamond, by Proposition~\ref{pULTI} we arrive at a contradiction. 
 \item $v' \in N(v_j)$ and $v'' \in N(v_k)$, for some $v_j, v_k \in V_2$, $j \neq k$.  Let $v'''$ be the other vertex adjacent to 
 $v_j$. Now, as $\{v_i,v',v_j\}$ induces a $P_3$, then by Theorem~\ref{tMarina}, $\{v_i,v',v_j,v'',v'''\}$ induces a gem (the only one) containing that $P_3$, since $v',v'''$ and $v'',v'''$ should be adjacent. Finally, if the vertices
 $v_i,v',v_j,v'',v''',v_k$ are respectively called $b,b_2,b',b_1,b_3,\tilde{b}$, then by Proposition~\ref{B3} we obtain a contradiction since $v_k$ cannot be adjacent to $v''$. Note that depending on the other neighbor $v''''$ of $v_k$ and their adjacencies, the \textit{Haj\'os graph} ($v'''' = v'''$), the \textit{rising sun} ($v''''$ adjacent to $v'$ and $v'''$) or the 
 $X_1$ graph ($v''''$ adjacent to $v'$ and not adjacent to $v'''$) appears.
\end{itemize}

Therefore, for each vertex of degree $2$, we can associate a unique neighbor of degree
greater than $2$ that is not adjacent to any other vertex of degree $2$. 
Thus, as there is at least one neighbor of a vertex of $V_2$ that is not associated to any vertex, and has degree greater than $2$, it follows that $|V_2| < n/2$.
\end{proof}

As an application of Theorem~\ref{teogrado2}, in Figure~\ref{FigExtraExtra3}, the first and third graphs have more vertices of degree $2$ than vertices of degree greater than $2$. Therefore, they are not biclique graphs. Other two examples are shown in Figure~\ref{ejteogrado2}.

\begin{figure}[ht!]
  \centering
  \includegraphics[scale=.3]{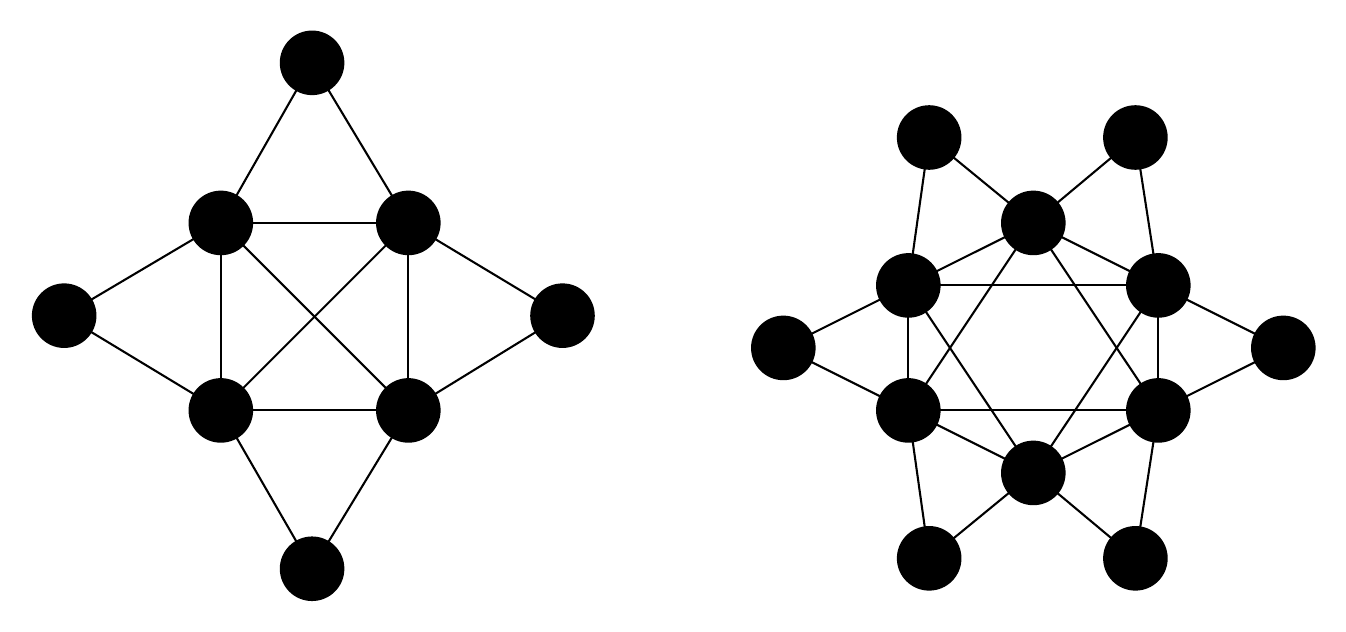} 
  \caption{Graphs that are not biclique graphs by Theorem~\ref{teogrado2}.}
  \label{ejteogrado2}
\end{figure}

As another application of Proposition~\ref{B3}, we obtain the following result.
For that, we need one more definition. A family of sets $\mathcal{A}$ is \textit{Helly} when 
every subfamily of pairwise intersecting subsets has a non-empty intersection.

\begin{proposition}\label{2-helly}
Let $G$ be a biclique graph and let $A = \{ N[v] : v \in G$ and $d(v) = 2\}$. Then $A$ is Helly.
\end{proposition}
\begin{proof}
Let, to the contrary, $A'$ be a minimal non-Helly subfamily of $A$. Now, since $A'$ is minimal and each $N[v_i] \in A'$ induces
a $K_3$ (as by Theorem~\ref{tMarina}, neighbors of $v_i$ should be adjacent), we have that $|A'|=3$. Moreover, as $A'$ is non-Helly, it induces the \textit{Haj\'os graph} where the vertices of degree $2$ in the subgraph are also of degree $2$ in $G$. Therefore by Corollary~\ref{coroHajos}, $G$ is not a biclique graph, which is a contradiction. We conclude that $A$ is a Helly family.
\end{proof}

\section{Open problems}

In this section, we present some conjectures. We look for proofs or counterexamples.

We propose first the following conjecture that generalizes Proposition~\ref{2-helly}.

\begin{conjecture}
 Let $G$ be a biclique graph and let $A = \{ N[v] : v\in G$ and $v$ is simplicial$\}$. Then $A$ is Helly.
\end{conjecture}

Proposition~\ref{pULTI} can be extended leading to the following conjecture.

\begin{conjecture}\label{nokb_general}
Let $G=KB(H)$ for some graph $H$, where $G$ is not isomorphic to the diamond. Then, there do not exist $v_{1},v_{2},\ldots,v_{i} \in V(G)$ such that 
$N{(v_{1})}=N{(v_{2})}=\ldots =N{(v_{i})}$ and their neighbors are a subgraph of $K_{i}$ for $i \geq 2$.
\end{conjecture}

\section{Conclusions}
In this work we give a formula for the distances between vertices in the biclique graph $KB(G)$ using the distances between bicliques of $G$. This is a useful tool for proving structural properties in bicliques graphs. In particular, it allows us to give a different proof for the necessary condition for a graph to be a biclique graph given 
in~\cite{GroshausSzwarcfiterJGT2010}. Also, it is used to answer (negatively) the question about the condition being sufficient or not.

Finally, we give an upper bound on the number of vertices of degree $2$ for biclique graphs. Also we give some forbidden structures that are useful to recognize graphs which are not biclique graphs.

\section*{Appendix}

In this section we present a complete list of biclique graphs up to $6$ vertices. For those graphs that verify the necessary condition of biclique graphs, we have checked by the computer whether they are biclique graphs using the characterization given in~\cite{GroshausSzwarcfiterJGT2010}.

\begin{figure}[ht!]
  \centering
  \includegraphics[scale=.27]{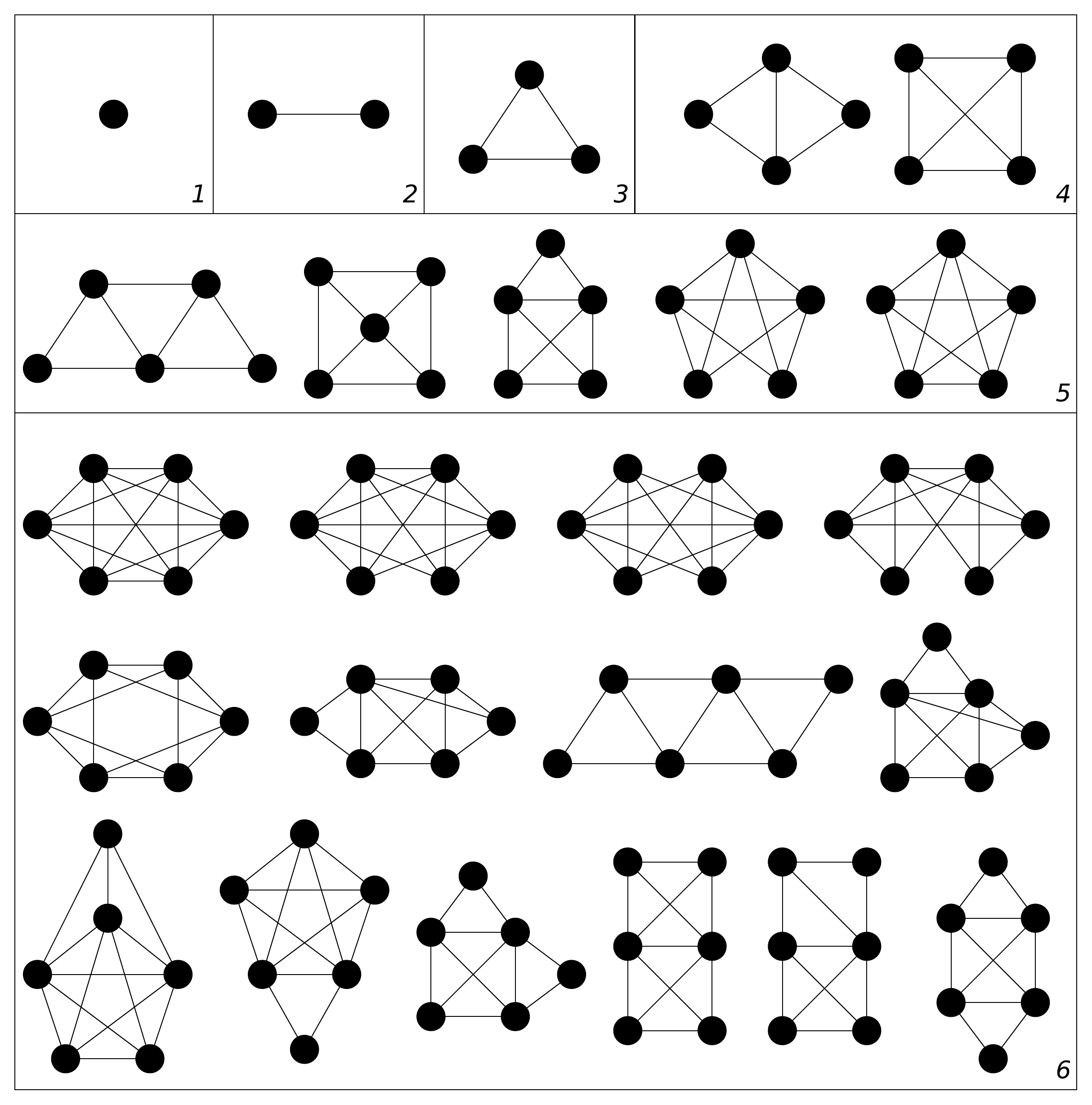} 
  \caption{Biclique graphs up to $6$ vertices.}
  \label{grafoskb}
\end{figure}

\bibliographystyle{abbrv}
\bibliography{biblio}

\end{document}